\documentclass[twoside,leqno,twocolumn]{article}
\usepackage{ltexpprt}

\usepackage{times}
\usepackage{xcolor}
\usepackage{soul}
\usepackage[utf8]{inputenc}
\usepackage[small]{caption}
\usepackage{multirow}
\usepackage{algorithm}
\usepackage{algorithmic}
\usepackage{graphicx}
\usepackage{bm}
\usepackage{enumitem}

\usepackage{amsmath,amssymb,amsfonts}
\usepackage{algorithmic}
\usepackage{graphicx}
\usepackage{textcomp}
\usepackage{xcolor}

\newcommand{\bfx}{{\bf x}}
\newcommand{\bfy}{{\bf y}}

\newcommand{\bfalpha}{\bm{\alpha}}
\newcommand{\bfeta}{\bm{\eta}}

\begin{document}
     
\title{\Large Change Point Detection for Compositional Multivariate Data 
}
\author{Prabuchandran K.J. \thanks{prabuchandra@iisc.ac.in, Amazon-IISc Postdoctoral Scholar.} \\
\and
Nitin Singh\thanks{nitisin1@in.ibm.com, IBM Research Labs, Bangalore, India.}
\and
Pankaj Dayama\thanks{pankajdayama@in.ibm.com, IBM Research Labs, Bangalore, India.}
\and
Vinayaka Pandit\thanks{pvinayak@in.ibm.com, IBM Research Labs, Bangalore, India.}
}
\date{}

\maketitle







\begin{abstract} \small\baselineskip=9pt
Change point detection algorithms have numerous applications in fields of scientific and economic importance.  We consider the problem of change point detection on compositional multivariate data (each sample is a probability mass function), which is a practically important sub-class of general multivariate data. While the problem of change-point detection is well studied in univariate setting, and there are few viable implementations for a general multivariate data, the existing methods do not perform well on compositional data. In this paper, we propose a parametric approach for change point detection in compositional data. Moreover, using simple transformations on data, we extend our approach to handle any general multivariate data. Experimentally, we show that our method performs significantly better on compositional data and is competitive on general data compared to the available state of the art implementations.
\end{abstract}


\section{Introduction}
Change point detection arises in a wide variety of applications like time series analysis \cite{chen2011parametric},  fault detection in industrial processes, segmentation of signals in bio-medical and seismic signal processing, dynamic social networks, online advertising and financial markets \cite{basseville1993detection}. In such applications, one is presented with a sequence of vector of observations and the goal is to identify the time points where the distributional properties of the observed data change. In many applications, however, we often have the multivariate data as sequence of probability mass functions (also known as {\em compositional data} \cite{grunwald1993time}). Examples include, the percentage revenue contribution of portfolio of products in a monthly sales data, the proportion of time spent by an equipment in different operating modes/efficiency band in a given time period etc.  Compositional time series analysis has been studied in literature in the context of forecasting, state space modeling \cite{snyder2017forecasting}. However, in compositional data our focus is on the change point detection problems.

The changes occuring in the distribution of the data sequence, based on the nature of applications, has been modeled as piecewise iid data models \cite{killick2012optimal}, linear or structural change models 
and Markov models \cite{montanez2015inertial}.
Further, based on specific assumptions on data like sparsity and high dimensionality \cite{wang2018high},  Graph based nature \cite{chen2015graph}, underlying subspace structure \cite{kawahara2007change}, functional nature \cite{berkes2009detecting} numerous change point detection techniques have been developed.  A comprehensive treatment of various data modeling approaches for change point detection along with different solution methodologies has been provided in \cite{aminikhanghahi2017survey} and \cite{truong2018review}.

There has been significant body of research on change point algorithms for the univariate (scalar) data e.g. Binary segmentation, Segment neighborhood, PELT \cite{killick2012optimal} and computationally efficient procedure like Group Fused Lasso \cite{bleakley2011group}. However, there are very few viable change point algorithms for general-purpose multivariate data  \cite{liu2013change}, \cite{matteson2014nonparametric}, \cite{wang2018high} and in particular, for compositional multivariate data. 

The change point algorithms can be dichotomized to fall under either Bayesian or non-Bayesian (frequentist) paradigms. Under the Bayesian paradigm, a prior distribution is assumed on the location of change points as well as on the parameters and then based on the observations appropriate posterior update is carried out for the location of change points \cite{adams2007bayesian}. However, one has to appropriately choose the prior distributions and if required, approximate the posterior computation. Under the non-Bayesian paradigm, the likelihood of the observations coming from one process is compared against its alternative piecewise segmented processes. A test statistic is computed based on some cost criterion like likelihood ratio (see Section \ref{dm}) and by comparing the test statistic with respect to a suitable threshold the better of the two process is chosen.  Note that choosing an appropriate threshold is still a open research problem \cite{chen2011parametric}. 

Another broad categorization of the change point detection algorithms is based on parametric or non-parametric modeling of the data. In the parametric setting \cite{chen2011parametric}, the data is assumed to come from a known family of distribution whose parameters are unknown and needs to be determined. The change point analysis then involves detecting changes in the parameters for the underlying family of distributions. The applicability of the parametric approach depends on the flexibility of the assumed distribution family.
 
In the non-parametric setting \cite{brodsky2013nonparametric}, it is assumed the data comes from piecewise segmented densities that do not necessarily have a specific form. Here, non-parametric density estimation tools are utilized to estimate the segmented densities and the likelihood test based on the ratio of densities is performed to identify change points. In \cite{kawahara2009change}, it is shown that it is often advantageous to estimate the ratio of the densities by a suitable model rather than estimating individual densities. In other non-parametric approaches \cite{matteson2014nonparametric}, the existence of density is not assumed instead the existence of certain moments is assumed. In \cite{matteson2014nonparametric},  two change point algorithms namely E-Divisive and E-Agglomerative for the multivariate time series data are proposed.  E-Divisive method in \cite{matteson2014nonparametric} is the state-of-art technique shown to perform well on multivariate data.
In \cite{wang2018high}, high dimensional time-series has been considered and  an algorithm (InspectChangePoint) for detecting sparse mean changes using sparse singular value decomposition on CUMSUM \cite{basseville1993detection} transformation of data matrix has been developed.
 
\subsection{Our contribution}
In this paper, we propose an online parametric Dirichlet change point algorithm (ODCP) when the observations form compositional data and subsequently extend our algorithm to unconstrained multivariate data.
 
Below are the salient features of our approach:
\begin{itemize}
\item We use Dirichlet parameterization for the compositional data as it is the natural parameterization under this setting. For the case of general multivariate data, we use a novel transformation of the data into the simplex (see Section \ref{mcd}) that preserves desirable statistical properties for efficient change point detection.
\item We follow the parametric likelihood approach as its underlying probabilistic framework is more robust to the scale of data  (our test statistic is based on the likelihood of data rather than Euclidean distance metrics on the data itself).
\item Our algorithm is fairly out-of-the-box and works reasonably without much parameter tuning.
\end{itemize}
 
Through extensive experimentation, we show ODCP to be the most viable change point detection method for compositional data.  Also, for general multivariate data our algorithm performs very well. In fact, it significantly outperforms the state of the art algorithms when the changes are predominantly in the variance while being competitive in other settings.
 
The rest of the paper is organized as follows. Section \ref{dm} describes our approach and presents the change point algorithm for compositional data. In Section \ref{gmv}, we extend our approach to general multivariate data. Numerical results on real world and synthetic data are presented in Section \ref{expCD}. Finally, Section \ref{conc} provides the concluding remarks.

\section{ Change Point Detection on Compositional Data}\label{dm}
In this section, we first describe the change point problem and then provide details of our solution for detecting multiple change points. 

\subsection{Change Point Problem}\label{cpp}
We assume the sequence of observations $\{\bfx_1,\bfx_2,\ldots,\bfx_T\}$, where
$\bfx_i \in \mathbb{R}^d, ~1\leq i\leq T$, could be partitioned into $k$ non-overlapping segments.  Within each partition, we assume the data is i.i.d coming from some underlying distribution with unknown parameters. The time instances  at which these delineations happen are called the change points. We denote these $k-1$ change points as $(\tau_1, \tau_2, \ldots, \tau_{k-1})$ where $0=\tau_0 < \tau_1 < \ldots < \tau_{k-1}$ . The goal in the change point problem is to determine the number of different segments in the data ($k$) as well as the locations of the  $k-1$ change points. 

\subsection{Our Solution Overview}\label{ov}


\begin{itemize}[leftmargin=*]
\item We model the data as coming from parameterized Dirichlet family of distributions (see Section \ref{pdm}). With this parameterization, the change points correspond to time instances with abrupt changes in Dirichlet parameters across them. 
\item We identify multiple change points by sequentially performing single change point detection. For this we maintain a window of data with the property that it contains atmost one change point. 
\item For single change point detection we use the standard hypothesis testing framework based on log-likelihood test statistic. 
\item We determine the likelihoods of data (see  \eqref{LLtau}, \eqref{LL0}) based on the maximum likelihood estimates of Dirichlet parameters \cite{minka2000estimating} 
\item We test the significance of the detected change point based on a random subset test (an adaptation of permutation test). 
\end{itemize}

\subsection{Data Modeling: Parametric Dirichlet model} \label{pdm}
Dirichlet distribution of order $d$, $d\geq 2$ with parameters $\bfalpha =
(\alpha_1,\ldots, \alpha_d)$ has probability distribution function given by
$$ p(\bfx | \bfalpha) = \frac{1}{B(\bfalpha)}\prod_{i=1}^d x_i^{\alpha_i-1},$$ where $$B(\bfalpha) = \frac{\prod_{i=1}^d \Gamma(\alpha_i)}{\Gamma(\sum_{i=1}^d \alpha_i)}$$
and $\bfx = (x_1,\ldots,x_d)$ with $x_i\geq 0$ and $\sum_{i=1}^d x_i = 1$.
As can be seen, the support of a dirichlet distribution of order $d$ is the
$(d-1)$-simplex $\Delta^{d-1} := \{\bfx\in \mathbb{R}^d: \sum_{i=1}^d x_i =
1,\, x_i\geq 0 \text{ for all } i=1,\ldots,d\}$.

For change point detection in compositional data (each
point lies on the $(d-1)$-simplex), we model the data to be generated from
family of dirichlet distributions with parameters $\bfalpha^{(1)},\ldots,
\bfalpha^{(k)}$. For $r=1,\ldots, k$, the samples ${\bfx_i},~\tau_{r-1} < i\leq \tau_r$,
are assumed to follow dirichlet distribution with parameters given by the
vector $\bfalpha^{(r)}$.  


\subsection{Multiple to Single Change Point Problem}\label{mts}
Having modeled our data, we approach the problem of identifying multiple
unknown change points by performing a sequence of single change point
detections. In our algorithm, we maintain an active window $\mathcal{A}$ of
observations on which we perform single change point detection. We initialize
the active window $\mathcal{A }$ to $\{\bfx_1,\ldots,\bfx_{I}\}$.
The initial size $I$ could be specified by the user based on the total number
of observations, or using prior knowledge of the process such that there is atmost one change point in any window of size $I$.
On this active window of data $\mathcal{A}$, we run our single change point
detection algorithm. If we discover any change point say at $\tau$, we prune
(discard) the data before the discovered change point and reset the active
data window to size $I$ starting from the discovered change point, i.e.,
$\mathcal{A}= \{\bfx_{\tau}, \ldots,\bfx_{\tau + I}\}$. On the other hand, if  we do
not discover any change point, we append new observations in batches to the active window, i.e., $\mathcal{A }=\mathcal{A} \cup
\mathcal{N}$ where $\mathcal{N} = \{\bfx_{I+1},\ldots, \bfx_{I+b}\}$ (data of batch size $b \geq 1$). We maintain the
ordering in the data even though we depict it by union operation. Note that for
$b=1$ our algorithm runs in completely online manner. 

 \subsection{Identifying single candidate change point}\label{scp}
We now describe the hypothesis testing framework  to detect a single change
point for the active window of observations $\mathcal{A}$. For ease of
exposition, we ignore the true indices in $\mathcal{A}$ and assume that
$\mathcal{A}=\{\bfx_i, ~i = 1,\ldots, t\}$. We test the null hypothesis $H_0$
(there is no changepoint) against the alternate hypotheses $H_\tau$ (there is a
changepoint at $\tau$) where $1 < \tau < t$:
 \begin{itemize} [leftmargin=*]
\item  $H_0:$ The data $\mathcal{A}$ comes from a single Dirichlet distribution.
 \item  $H_{\tau}$: The data $\mathcal{A}$ comes from two Dirichlet distributions delineated at some $\tau$ where $1 < \tau < t$.
 \end{itemize}
Under hypothesis $H_0,$  we obtain the best parameter $\bfeta_0$ that explains
the data coming from single Dirichlet distribution by performing the maximum
likelihood estimation (MLE, see \cite{minka2000estimating}) on the complete data
$\mathcal{A}$. For each $\tau$, under the alternative hypothesis $H_\tau$,  we
perform two MLE estimations for the observations denoted as left data
$\mathcal{L}=\{\bfx_i, ~1\leq i\leq \tau\}$ and the observations
denoted as right data  $\mathcal{R}=\{\bfx_i,  \tau+1\leq i \leq t\}$. We then obtain the Dirichlet MLE estimates $\bfeta_{\mathcal{L}}^{\tau}$
and $\bfeta_{\mathcal{R}}^{\tau}$ for the left and right data
respectively. Now, for each $\tau$, we compute the log likelihood $LL_\tau$ (see \eqref{LLtau}) of the data
$\mathcal{A}$ under hypothesis $H_\tau$. Finally we choose $\tau^*$ as the
value of $\tau$ that maximizes the log likelihood of the data, i.e., 
\begin{eqnarray} \label{tauStar}
\tau^* &= \arg\max_{\tau} LL_{\tau}
\end{eqnarray}
where $LL_{\tau}$ is given by:
\begin{equation}\label{LLtau}
\resizebox{0.86\columnwidth}{!}
{
$LL_{\tau}  =  \ln ~p(\bfx_1, \ldots, \bfx_{\tau} \mid \bfeta_{\mathcal{L}}^{\tau}) + \ln ~p(\bfx_{\tau +1}, \ldots, \bfx_t \mid \bfeta_{\mathcal{R}}^{\tau} )$ 
}
\end{equation}
The log likelihood of the data under $H_{\tau^*}$ is then given by $LL_{\tau^*}$. Similarly, we compute the log likelihood of the data under $H_{0}$ as
\begin{eqnarray}\label{LL0}
LL_{0} &=  \ln ~p(\bfx_{1}, \ldots, \bfx_{t} \mid \bfeta_{\mathcal{A}} ) 
\end{eqnarray}
After determining the log likelihoods under $H_{\tau^*}$ and $H_0$ based on
equations \eqref{LLtau} and \eqref{LL0}, we decide if there is a change point
at $\tau^*$  by considering the log-likelihood ratio $Z^* := LL_{\tau*}-LL_0$ and deciding if $Z^\ast$ is statistically significant. For this, we propose a random subset test (a faster version of the statistical permutation test). 

\subsection{Significance test}\label{ssb:sf}
Having obtained the test statistic $Z^\ast$ for active window ${\cal A}$, we
need to determine its significance. We follow the commonly used permutation
test to determine the significance of the candidate test statistic $Z^\ast$.  
For this, we consider $M$ random permutations ${\cal A}_1,\ldots, {\cal A}_M$ of the active observation
window $\mathcal{A}$ and for each $i=1,\ldots,M$, compute the
test statistic $Z^\ast_i$ (as described in Section \ref{scp}) for each 
${\cal A}_i$. For a chosen significance level $\alpha$, we
reject the null hypothesis $H_0$ if the value $Z^\ast$ lies within the top
$\alpha$ fraction of the values $\{Z^\ast_i: i=1,\ldots,M\}$. Typically $\alpha$
is chosen to be a small number like $0.05$. 

 
Note that by considering finite number of permutations, the fraction
$\frac{|\{i:Z_i^* \geq Z^*\}|}{M}$ gives only an approximate p-value. In order
to get exact p-value, one needs to consider all permutations of ${\cal A}$,
however, this is computationally infeasible. Even considering only finite number of permutations to test the significance of a change point is a computational burden as this leads to computational complexity of O($|\mathcal{A}|$*M) each time we do the significance test.  We observe that,  under the independent samples assumption, the computation of the likelihood ratio is invariant to permutations of left and right data considered in the computation of $LL_{\tau}$ (see \eqref{LLtau}). Thus, instead of considering the random permutations and then partitioning the data as left and right data, we directly consider random subsets of the data as left partition and their complements in ${\cal A}$ as corresponding right partitions. 
\section{Change Point Detection On General Multivariate Data}\label{gmv}

In this section, we describe how our algorithm ODCP could be leveraged for non-compositional data. 
\subsection{Our Solution Overview}
\begin{itemize}[leftmargin=*]
\item We map the general multivariate data to compositional data using transformations that preserve the likelihood based test statistic, permutation distribution of the test statistic and the KL divergence between the distributions separated by change points (see Lemmas \ref{kl}, \ref{llr} ).
\item The properties of the transformations allow us to justifiably determine change points on the transformed compositional data.
\end{itemize} 

\subsection{Mapping to Compositional Data}\label{mcd}
Let $\mathcal{Y}=\{\bfy_1,\bfy_2,\ldots,\bfy_T\}$, $\bfy_i \in \mathbb{R}^d$ for  $1\leq i \leq T$ be general multivariate data. We transform $\mathcal{Y}$ as follows:
In the first step, we mean shift and scale by the standard deviation using the map ${\bf f}^1 :  \mathbb{R}^d \rightarrow \mathbb{R}^d$ as ${\bf f}^1(\bfy) = \frac{\bfy-\bm{\mu}}{\bm{\sigma}}$ where $\bm{\mu}$ and $\bm{\sigma}$ denote to the vector of component-wise mean and standard deviation of the data $\mathcal{Y}$ respectively.  Here the division of vectors is performed component-wise. In the next step, we map the mean centered and scaled data into the simplex by using multi-dimensional expit function (inverse of multinomial logit function) ${\bf f}^2 :  \mathbb{R}^d \rightarrow interior(\Delta^d)$ given by:
\begin{equation}\label{invLogit}
\resizebox{0.86\columnwidth}{!}
 {
 ${\bf f}^2(\bfy) = \left[{\frac {e^{y_1}}{1+\sum _{i=1}^{d}e^{y_i}}},\dots ,{\frac {e^{y_d}}{1+\sum _{i=1}^{d}e^{y_i}}},{\frac {1}{1+\sum _{i=1}^{d}e^{y_i}}}\right]^{\top }$
 }
 \end{equation}
Let $\mathcal{X}=\{\bfx_1,\bfx_2,\ldots,\bfx_T\}$ where $\bfx_i = f^2(f^1(\bfy_i)), ~ i \in \{1, 2, \ldots, T\}$. Note the transformed data $\mathcal{X}$ is compositional as it lies on the simplex. 

\subsection{Justification for our Transformation}
In this section, we consider the single change point detection problem for the given data $\mathcal{Y}$ and prove the invariance of statistical properties of $\mathcal{Y}$ under the transformations. This conveniently enables us to work on the transformed data $\mathcal{X}$.  We now begin with the assumption that the data $\mathcal{Y}$ is generated by either of the following process:
\begin{itemize}[leftmargin=*]
\item  $H_0:$  $\{\bfy_1,\ldots,\bfy_T\}$ are i.i.d samples from distribution with probability density $p_0$.
 \item  $H_{\tau}$: $\{\bfy_1,\ldots,\bfy_\tau\}$ are i.i.d samples from distribution with density $p_1^{\tau}$ while $\{\bfy_{\tau+1},\ldots,\bfy_T\}$ are i.i.d samples from distribution with density $p_2^{\tau}$. 
 \end{itemize}
Under the hypothesis $H_{\tau}$, let $p_{m}^{\tau}$ denote the mixture distribution on data $\mathcal{Y}$ arising from the $\tau$ left samples from $p_1^{\tau}$ and $T - \tau$ right samples from $p_2^{\tau}$.

Let $ \textbf{h} = {\bf f}^2 \circ {\bf f}^1 : \mathcal{Y} \subset \mathbb{R}^d \rightarrow \mathcal{X} \subset interior(\Delta^d)$ denote the composed forward transformation and $\textbf{g} =\textbf{h}^{-1} : \mathcal{X} \rightarrow \mathcal{Y}$  denote its inverse transformation
and let $\textbf{J} = \frac{\partial (y_1,\ldots, y_d)}{\partial (x_1,\ldots, x_d)} $ denote the Jacobian of the inverse transformation.


Now the following holds true for data ${\mathcal X}$ under respective hypothesis:
\begin{itemize}[leftmargin=*]
\item  $H_0:$  $\{\bfx_1,\ldots, \bfx_T\}$  are  i.i.d samples from distribution with probability density $q_0$ where 
\begin{eqnarray}\label{q0}
q_0(\bfx) = p_0(\textbf{g}(\bfx)) | {\rm det}\,  \textbf{J}(\bfx)|
\end{eqnarray}
 \item  $H_{\tau}$: $\{\bfx_1,\ldots,\bfx_\tau\}$ are i.i.d samples from distribution with probability density $q_1^{\tau}$ where
 \begin{eqnarray}\label{q1}
  q_1^{\tau}(\bfx) = p_1^{\tau}(\textbf{g}(\bfx)) | {\rm det}\,  \textbf{J}(\bfx)| 
 \end{eqnarray}
  while $\{\bfx_{\tau+1},\ldots,\bfx_T\}$ are i.i.d samples from distribution with probability density $q_2^{\tau}$ where
 \begin{eqnarray}\label{q2}
   q_2^{\tau}(\bfx) = p_2^{\tau}(\textbf{g}(\bfx)) | {\rm det}\,  \textbf{J}(\bfx)| 
  \end{eqnarray} 
\end{itemize}
Under the hypothesis $H_{\tau}$, let $q_{m}^{\tau}$ denote the mixture distribution on data $\mathcal{X}$ arising from the left $\tau$ samples from $q_1^{\tau}$ and right $T - \tau$ samples from $q_2^{\tau}$ 
\begin{eqnarray}\label{qm}
 q_m^{\tau}(\bfx) = p_m^{\tau}(\textbf{g}(\bfx)) | {\rm det}\,  \textbf{J}(\bfx)| 
 \end{eqnarray}
We now have the following lemma:

\begin{lemma}\label{kl}
KL($p_m^{\tau}||p_0)$ = KL($q_m^{\tau}||q_0)$
\end{lemma}
 \begin{proof}
 \begin{eqnarray*}
KL(p_m^{\tau}||p_0) & = & \int p_m^{\tau}(\bfy) \ln \frac{p_m^{\tau}(\bfy)}{p_0(\bfy)} d\bfy \nonumber \\
&=& \int p_m^{\tau}(\textbf{g}(\bfx)) \ln \frac{p_m^{\tau}(\textbf{g}(\bfx))}{p_0(\textbf{g}(\bfx))} |{\rm det}\,  \textbf{J}(\bfx)| d\bfx \nonumber  \\
& & (\text{change of variables with} ~\bfy=\textbf{g}(\bfx)) \nonumber \\
&=& \int q_m^{\tau}(\bfx) \ln \frac{q_m^{\tau}(\bfx)}{q_0(\bfx)} d\bfx \nonumber \\
&=& KL(q_m^{\tau}||q_0) \nonumber  
\end{eqnarray*}
 \end{proof}
Lemma \ref{kl} essentially states the KL distance between the mixture distribution under $H_{\tau}$ and the distribution under $H_0$ is preserved  by the transformation $\textbf{h}$. Under the likelihood framework, we have a stronger result that the transformation $\textbf{h}$ preserves positions of change points in addition to the KL distances. For this, let us define the log-likelihood ratio $LLR^{\tau}$ for data $\mathcal{Y}$ with possible change point at $\tau$ as
\begin{eqnarray}
LLR^{\tau}(\mathcal{Y}) = \ln \frac{p(\mathcal{Y} | H_{\tau})}{p(\mathcal{Y} | H_0)}
\end{eqnarray}
 \begin{lemma}\label{llr}
$LLR^{\tau}(\mathcal{Y}) = LLR^{\tau}(\mathcal{X})$
\end{lemma}
 \begin{proof}
 \begin{eqnarray}
LLR^{\tau}(\mathcal{Y}) &=& \ln\Big(\frac{\prod_{i=1}^{\tau} p_1^{\tau}(\bfy_i) \prod_{i=\tau+1}^{T} p_2^{\tau}(\bfy_i)}{\prod_{i=1}^{T}{p_0^{\tau}(\bfy)}}\Big) \nonumber \\
&=& \ln\Big(\frac{\prod_{i=1}^{\tau} q_1^{\tau}(\bfx_i) \prod_{i=\tau+1}^{T} q_2^{\tau}(\bfx_i)}{\prod_{i=1}^{T}{q_0^{\tau}(\bfx_i)}}\Big) \nonumber \\
& & \resizebox{0.7\columnwidth}{!}{($\text{easily follows from} ~\eqref{q0}, \eqref{q1}, \eqref{q2} ~\text{and} ~\eqref{qm}) \nonumber $} \\
&=& LLR^{\tau}(\mathcal{X}) \nonumber  
\end{eqnarray}
 \end{proof}

Lemma \ref{llr} establishes the log-likeihood ratios are preserved under the transformation $\textbf{h}$.  Moreover, as the permutation test that determines statistical significance of the selected change point also depends on the log-likelihood ratios of the permuted data, the statistical testing is also invariant under the transformation. With these statistical invariants, it is thus sufficient to perform log-likelihood based change point detection on data $\mathcal{X}$ instead on $\mathcal{Y}$.  This offers us significant advantage as we have flexibility to model the distribution on the simplex using Dirichlet family.
The Dirichlet parameterization has other advantages as follows:
\begin{itemize}[leftmargin=*]
\item Dirichlet Parameter estimation is a convex optimization problem leading to reliable and fast MLE estimates.
\item The number of parameters to be estimated is linear in dimension of the data unlike quadratic in the case of Gaussian parameter estimation.
\item Although the number of parameters to be estimated are less, they still capture the interaction between different component variables (as each component of the mean of Dirichlet involves all other Dirichlet parameters).
\end{itemize}  

\section{Experiments}\label{expCD}
In this section, we perform extensive simulations comparing our algorithm ODCP against the state of the art algorithms. We divide the set of our experiments into two parts based on the kind of data considered. Our first set of experiments are on compositional data (Section  \ref{comp}). Here we show the limitations of existing multivariate change point methods to identify change points and simultaneously show the efficacy of our approach on varied range of compositional data. In the subsequent set of experiments, we test the extension of our approach on non-compositional data (Section \ref{noncomp}). In both the compositional and non-compositional data settings that we considered, we perform experiments on synthetic and real datasets.

\subsection{Preliminary Experiments}\label{prelim}
Although there are numerous change point algorithms in literature, for empirical evaluation we  could only consider algorithms that have available implementations. Further, only few of the available implementations are applicable to multivariate change point detection. In this experiment, we determine viable/robust change point algorithms based on the datasets  provided in \cite{erdman2007bcp}, \cite{matteson2014nonparametric} and \cite{wang2018high}. Based on the performance of different algorithms on these datasets, we shortlist the most promising algorithms for further extensive comparison in later sections. We consider five datasets $\mathcal{S}_{1}, \ldots, \mathcal{S}_{5}$. We now briefly describe them (see \cite{james2013ecp} \cite{erdman2007bcp} for more details).

\noindent {\it Change in mean and variance (${\mathcal S}_{1}$)}:  This is a univariate data where 100 independent samples were generated from the following normal distributions: $\mathcal{N}(0,1), \mathcal{N}(0,\sqrt{3}), \mathcal{N}(2,1), \mathcal{N}(2,2)$. This data contains three change points. The first is due to change in only variance, second is due to change in both mean and variance and the third is due to change in only variance. 

\noindent {\it Change in Covariance (${\mathcal S}_{2}$)}: This is a multivariate data with dimension $d=3$ where 250 independent samples were generated from the following normal distributions: $\mathcal{N}(0,\Sigma_1), \mathcal{N}(0,\Sigma_2), \mathcal{N}(0,\Sigma_1)$, where 
\begin{equation*}
\Sigma_1 = \begin{pmatrix}
1 & 0 & 0 \\
0 & 1 & 0 \\
0 & 0 & 1
\end{pmatrix} \quad
\Sigma_2 = \begin{pmatrix}
1 & 0.9 & 0.9 \\
0.9 & 1 & 0.9 \\
0.9 & 0.9 & 1
\end{pmatrix}
\end{equation*}
This data has two change points and both are due to only change in covariance. \smallskip

\noindent {\it Change in tail behaviour (${\mathcal S}_{3}$)}: This data has two change points and both are due to change in tail behavior. Here, data points are drawn from a bivariate normal distribution and a bivariate t distribution with 2 degrees of freedom (see \cite{matteson2014nonparametric} for more details). \smallskip

\noindent {\it Change in Poisson parameter intensity (${\mathcal S}_{4}$)}: This data has three change points corresponding to change in the intensity parameter of a Poisson process (see \cite{matteson2014nonparametric} for more details). \smallskip

\noindent {\it Sparse Change in Mean with Overlap (${\mathcal S}_{5}$)}: This is a multivariate data with dimension $d=50$ with sparse change in mean in some of the components (adaptation of  example data from \cite{wang2018high})\smallskip

\noindent \textbf{Results for Preliminary Experiment}:  In Table \ref{genData}, we compare ODCP with other available state of the art change point methods namely ECP (E-Divisive from \cite{matteson2014nonparametric}), Bayesian Change Point Detection (BCP) \cite{adams2007bayesian}, Double CUSUM \cite{cho2016change}  and InspectChangePoint  (InsCP) \cite{wang2018high} on datasets $\mathcal{S}_{i}, ~ i \in \{1,\ldots,5\}$. From Table \ref{genData}, we observe that only ODCP, ECP and InsCP are able to effectively identify change points.  However, other algorithms are not as competitive. Therefore, we compare ODCP with two most promising methods, namely,  ECP and InsCP for further evaluation over varied datasets in subsequent sections.
\begin{table}[h]
\begin{center}
\begin{tabular}{|c|c|c|c|c|c|c|}
\hline
Dataset & ODCP & ECP & BCP & DCSUM & InsCP\\ \hline
${\mathcal S}_{1}$ &  3/3, 0 & 3/3, 0  & 2/3, 1  & 1/3, 0 & 1/3, 0\\
${\mathcal S}_{2}$ &  2/2, 0 & 2/2, 0 & 0/2, 2 & $-$ & 0/2,0\\
${\mathcal S}_{3}$ &  2/2, 0 & 2/2, 0 & 0/3, 3 & 0/3, 1 & 0/2,0\\
${\mathcal S}_{4}$ &  3/3, 0 & 3/3, 0 & 1/3, 0 & $-$ & 3/3,0\\
${\mathcal S}_{5}$ &  2/3, 0 & 3/3, 0 & 0/3, 0 & 0/3, 0 & 3/3,0\\
\hline
\end{tabular}
\caption{\small{General data available from literature. Tuple Entries $a/b,c$ denote $a$ out of $b$ true change points were detected by the algorithm while there were $c$ false positives. $-$ entry denotes that the algorithm failed to yield results}}\label{genData}
\end{center}
\end{table}

\subsection{Compositional Data}\label{comp}
\subsubsection{\bf Experimental Setup} \label{expSetupComp}
In the following set of experiments, we consider a single change point detection problem on 10-dimensional data ($d =10$) with total of 1000 observations ($T=1000$). We set the location of the change point  at $\tau=500$, i.e., we consider two segments of length $500$ where the data within each segment is homogeneously generated from suitable distribution. Based on the chosen distribution, we generate four different datasets (compositional data) as described below:\smallskip

\noindent{\it Dirchlet  Data (${\cal D}_1$)}: This dataset is generated by sampling from Dirichlet distribution. Different Dirichlet parameters are chosen for each segment. These parameters are chosen such that the symmetric KL distance between the distribution across segments is about 0.5 (this intuitively captures the ``closeness'' of the two segments). \smallskip

\noindent{\it Dirchlet Mixture Data (${\cal D}_2$)}: This dataset is generated by sampling from Dirichlet mixture distribution with three components.  In both the segments, we keep the mixture proportion constant at $(0.3,0.4,0.3)$ while between the segments the parameter of the component Dirichlet distributions are perturbed slightly. \smallskip

\noindent{\it Gaussian Normalized Data (${\cal D}_3$)}: To generate this dataset, we first generate data from 10-dimensional Gaussian distributions. Between segments, slightly different parameters are chosen. After sampling Gaussian data, we finally apply the transformation ${\bf x}\rightarrow {\bf x}/{|{\bf x}|}_1$ to this dataset to obtain compositional data.\smallskip

\noindent{\it Logistic Normal Data (${\cal D}_4$)}: To generate this dataset, we first generate data from 10-dimensional Gaussian distribution (as in previous case). The segments are differentiated by choosing slightly different parameters for each of them. Subsequently, we obtain compositional data by projecting this dataset onto the simplex by applying inverse multinomial logit transformation.

\subsubsection{\bf Description of Comparison Algorithms}
In this section, we briefly describe the algorithms that we compare against ODCP. The candidate algorithms are chosen based on our preliminary experiments described in Section \ref{prelim}.
\begin{itemize}
\item ECP is a non-parametric algorithm that uses a classical divergence measure between two distribution based on characteristic functions and the energy statistic \cite{matteson2014nonparametric}. The implementation of the algorithm is available through {\tt ecp} package in R.
\item InsCP is also a non-parametric algorithm that detects sparse changes in the mean structure in high dimensional data \cite{wang2018high}. The implementation of the algorithm is available through {\tt InspectChangePoint} package in R.
\end{itemize}

\subsubsection{\bf Performance Metrics}
Labelling a data point as change point is said to be correct if it lies within the predefined window size $W$ of the actual change point.  Based on this criteria, for each of the algorithms, we compute the standard precision and recall metric. The default window size for evaluation is set to $4\%$ of segment length (typically, 500 in our experiments).  We also compare the performance of the algorithms for different window sizes using precision and recall curves (see Figures \ref{meanSparsityComp}, \ref{varianceSparsityComp} in Section \ref{sparseData}). 
For each of the experiments, the performance metric is obtained by averaging over 20 Monte-Carlo simulations, i.e., we sample 20 datasets randomly for performing Monte-Carlo averaging.
\subsubsection{\bf Discussion of Results} 
In this section, we present the results for the three algorithms on each of the four synthetically generated compositional datasets described in Section \ref{expSetupComp}.\medskip

\noindent{\it Dirchlet  Data (${\cal D}_1$)}: From first row of Table \ref{DirMIx}, it can be seen that ODCP effectively detects change points when the data comes from Dirichlet family. This is to be expected as we perform Dirichlet parameter estimations in our algorithm. \smallskip

\noindent{\it Dirchlet Mixture Data (${\cal D}_2$)}: The effectiveness of ODCP degrades in distinguishing Dirichlet mixture distribution (see second row of Table \ref{DirMIx}). On the other hand, other algorithms' performance is poor in both the settings.\smallskip 
\begin{table}[h]
\begin{center}
\resizebox{0.8\columnwidth}{!}{
\begin{tabular}{|c|c|c|c|c|c|c|}\hline
Datasets &  \multicolumn{2}{|c|}{ODCP} & \multicolumn{2}{|c|}{ECP} &    \multicolumn{2}{|c|}{InsCP} \\
   \cline{2-7}
   & Prec. & Rec. & Prec. & Rec. & Prec. & Rec \\
   \cline{1-7}
$\mathcal{D}_1$ &   0.8 & 1.0 & 0.5 & 0.25 &1.0 &0.05 \\
\hline
$\mathcal{D}_2 $ &  1.0 & 0.6 & $-$ & 0.0 & $-$ & 0.0\\ \hline
\end{tabular}}
\caption{\small{Performance comparison on Dirichlet $\mathcal{D}_1$ and Mixture Data $\mathcal{D}_2$ with window size 4\%. $-$ indicates the algorithm did not detect any change point.}}\label{DirMIx}
\end{center}
\end{table}

\noindent{\it Gaussian Normalized Data (${\cal D}_3$)}: 
Since, datasets in ${\cal D}_3$ are obtained from Gaussian family, we introduce change points in data by choosing different gaussian parameters for each segment. 
Specifically we consider the following variations:
\begin{itemize}[leftmargin=*]
\item Type of change: change in mean, change in variance.
\item SNR: We consider varying signal to noise ratio from 5 (High SNR, low noise)  to 20  (Low SNR, high noise) for the base signal.
\end{itemize}

The results of the experiments on this dataset are summarised in Table \ref{normalData}. In the case of mean change, ODCP identifies almost all the change points (high recall) under high SNR. Under low SNR, the performance of all the methods degrades.  However, in the case of variance change, ODCP performs significantly better than others irrespective of the noise in the data. \smallskip


\begin{table}[h]
\begin{center}
\resizebox{0.8\columnwidth}{!}{
  \begin{tabular}{|l|l|l|l|l|l|l|l|}
    \hline
Type & SNR &   \multicolumn{2}{|c|}{ODCP} & \multicolumn{2}{|c|}{ECP} &    \multicolumn{2}{|c|}{InsCP} \\
         \cline{3-8}
  &  & Prec. & Rec. & Prec. & Rec. & Prec. & Rec \\
 \hline
\multirow{2}{*}{Mean} & high & 0.89 & 0.85 &  0.65  & 0.81 & 1.0 & 0.64 \\
\cline{2-8}
 &  low & 0.69 & 0.45  & 0.77 &   0.35  & 0.96 & 0.50 \\
 \hline
\multirow{2}{*}{Var} & high & 0.95 & 1.0 & 0.85   & 0.60 & $-$ & $-$ \\
\cline{2-8}
 &  low & 0.81 & 0.65  & 0.78 &  0.55  & $-$ & $-$ \\
 \hline
  \end{tabular}}
  \caption{\small{Performance comparison on variants of normalized Gaussian data $\mathcal{D}_3$} with window size 4\%. '-' indicates that the algorithm is not applicable}\label{normalData}
  \end{center}
\end{table}

\noindent{\it Logistic Normal Data (${\cal D}_4$)}: This dataset differs from $\mathcal{D}_3$ in the kind of transformation that is applied onto the Gaussian data. Similar to dataset $\mathcal{D}_3$,  we consider changes in mean and variance under different SNR.
From Table \ref{logitData}, we observe that ODCP substantially outperforms ECP when there is change only in variance. InsCP performs the best in the case of detecting change in mean. However, ODCP and ECP are also very competitive especially in the case of low noise (High SNR).
\smallskip

\begin{table}[h]
\begin{center}
\resizebox{0.8\columnwidth}{!}{
  \begin{tabular}{|l|l|l|l|l|l|l|l|}
    \hline
Type & SNR &   \multicolumn{2}{|c|}{ODCP} & \multicolumn{2}{|c|}{ECP} &    \multicolumn{2}{|c|}{InsCP} \\
         \cline{3-8}
  &  & Prec. & Rec. & Prec. & Rec. & Prec. & Rec \\
 \hline
\multirow{2}{*}{Mean} & high & 0.90 & 0.95 &  0.88  & 0.80 & 1.0 & 1.0 \\
\cline{2-8}
 & low & 0.89 & 0.85  & 0.82 &   0.7  & 1.0 & 1.0 \\
 \hline
\multirow{2}{*}{Var} & high & 0.95 & 1.0 & 0.85   & 0.81 & $-$ & $-$ \\
\cline{2-8}
 &  low & 0.90 & 1.0  & 0.78 &  0.55  & $-$ & $-$ \\
 \hline
  \end{tabular}}
  \caption{\small{Performance comparison on variants on logistic normal data $\mathcal{D}_4$} with window size 4\%}\label{logitData}
  \end{center}
\end{table}

\noindent{\bf Summary of Results on Synthetic Data}:
\begin{itemize}[leftmargin=*]
\item ODCP consistently performs well overall on all the compositional datasets we have considered. In particular, it outperforms other algorithms in identifying changes in variance.
\item Thus, ODCP appears to be the most viable method for effective change point detection on compositional data.
\end{itemize}

\subsubsection{\bf  Experiments for studying effects of sparsity}\label{sparseData} 
In addition to the earlier experiments, we investigate the effectiveness of ODCP when changes occur in sparse subset of the coordinates. Note that this is a plausible scenario in practice. In this experiment, we generate data with dimension $d=10$ and total number of observations $T=1200$. The data consists of four segments each of length $300$, i.e,, the true change points are located at 300, 600 and 900. Each segment is obtained by sampling from a 10-dimensional Gaussian distribution and transforming to compositional data as in $\mathcal{D}_3$. As before, we introduce changes in mean and variance between different segments but with a specified sparsity level. For sparsity level $s,  ~0  \leq s \leq 1$, we restrict the changes (mean/variance) to randomly chosen $s$ fraction of the coordinates. We consider three specific sparsity levels corresponding to $s=\{0.3,0.5,0.7\}$. For each of the algorithms, we measure the precision and recall metrics for different window sizes $W$ where $ W \in \{0, 1, 2, \ldots, 50\} $. Since the segment length is $300$, this corresponds to detection error tolerance of  $0 \%$ to $16\%$. The precision and recall curves for the three algorithms in detecting changes in mean and variance are depicted in Figures \ref{meanSparsityComp} and \ref{varianceSparsityComp} respectively. In the case of sparse mean change (Figure \ref{meanSparsityComp}), all the three algorithms perform equally well. However, in the case of sparse variance change (Figure \ref{varianceSparsityComp}), ODCP significantly outperforms ECP.  The recall metric for ECP does not exceed 0.6 under all sparsity levels, however, the recall for ODCP significantly improves. 
 Note that the InsCP algorithm is not applicable here as it detects only sparse changes in mean (when we tried InsCP in this experiment, it did not detect any change point). 

\begin{figure}[ht!]
\centering
\includegraphics[width=0.4\textwidth, height=0.4\textwidth]{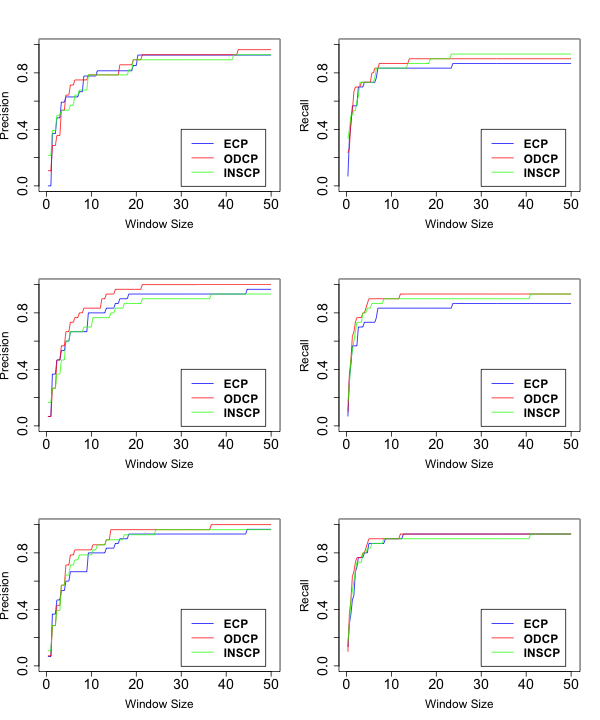}
\caption{\small Effect on performance for sparse mean change on compositional data. Results for sparsity levels of $0.3,0.5$ and $0.7$ are shown in rows first, second and third respectively.}
\label{meanSparsityComp}
\end{figure}

\begin{figure}[ht!]
\centering
\includegraphics[width=0.4\textwidth, height=0.4\textwidth]{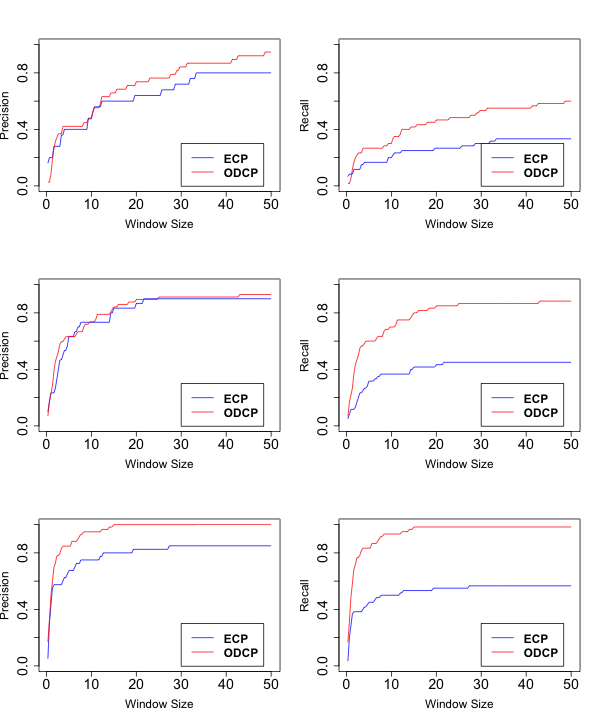}
\caption{\small Effect on performance for sparse variance change on compositional data. Results for sparsity levels of $0.3,0.5$ and $0.7$ are shown in rows first, second and third respectively.}
\label{varianceSparsityComp}
\end{figure}

\subsubsection{\bf Discussion of Results for Real Data}\label{realComp} 
In this section, we use our algorithm to obtain insights about significant process changes in a real world compositional data from an industrial plant. 
The dataset describes the proportion of time spent by an equipment (cement mill) in four different modes of operation, each corresponding to a different energy efficiency band, in a day.  The dataset had records for 230 days. Using our approach, we identified 3 change points. We correlated these change points with the time-series data being measured for the equipment using sensors for different process variables such as table vibration, classifier load, fan speed, etc. Most of the change points lead to key insights in understanding changes in machine operation. In Fig. \ref{cprealdata} we plot the important process variables and clearly indicate the change points identified by ODCP.  The plots show that the change points correlate closely with changes in these important variables.

\begin{figure}[ht!]
\centering
\includegraphics[width=0.4\textwidth,height=0.2\textwidth]{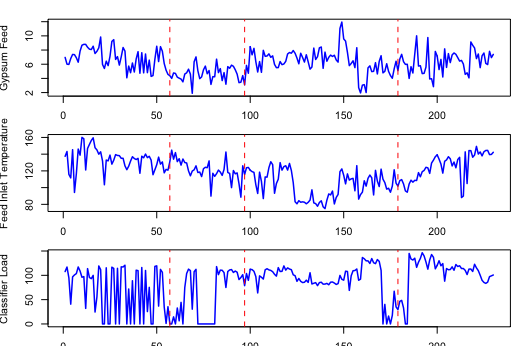}
\caption{\small Correlation of the change points with important process variables for the equipment. First change point correlates with mean shift of gypsum feed and change in variance of classifier load. Second changepoint shows strong correlation with mean shift in gypsum feed and weak correlation with change in feed inlet temperature. Last change point strongly correlates to the increasing trend in feed inlet temperature. }
\label{cprealdata}
\end{figure}

\subsection{Non-Compositional Data}\label{noncomp}
In what follows, we evaluate the performance of ODCP, ECP and InsCP in the general (non-compositional) data setting.

\subsubsection{\bf Experimental setup}
The setup in this experiment is similar to that in the compositional case (see Section \ref{expSetupComp}). We generate non-compositional data as described below:\smallskip

\noindent{\it Gaussian Data (${\cal G}_1$)}: This dataset is generated in a similar manner to datasets $\mathcal{D}_3$ and $\mathcal{D}_4$. However, for this dataset we do not apply  the transformation (normalization or multinomial inverse logit) for converting the  multivariate Gaussian data to compositional data.

\subsubsection {\bf Discussion of Results}
The results for non-compositional setting are tabulated in Table \ref{gData}. 
From the table, we observe that all the three algorithms are equally good in identifying change-points in most scenarios we have considered. ODCP outperforms ECP when there is change only in variance. This is clearly illustrated by precision and recall curves in Figure \ref{varianceSparsityNonComp}. 
\smallskip
 
  \begin{table}[h]
\begin{center}
\resizebox{0.8\columnwidth}{!}{
  \begin{tabular}{|l|l|l|l|l|l|l|l|}
    \hline
Type & SNR &   \multicolumn{2}{|c|}{ODCP} & \multicolumn{2}{|c|}{ECP} &    \multicolumn{2}{|c|}{InsCP} \\
         \cline{3-8}
  &  & Prec. & Rec. & Prec. & Rec. & Prec. & Rec \\
 \hline
\multirow{2}{*}{Mean} & high & 1.0 & 1.0 &  1.0  & 0.83 & 1.0 & 1.0 \\
\cline{2-8}
 & low & 1.0 & 1.0  & 0.83  &  1.0  & 1.0 & 1.0 \\
 \hline
\multirow{2}{*}{Var} & high & 1.0 & 1.0 & 0.8   & 0.81 & $-$ & $-$ \\
\cline{2-8}
 &  low & 0.9 & 0.9  & 0.83 &  0.5  & $-$ & $-$ \\
 \hline
  \end{tabular}}
\caption{\small{Performance comparison on variants on Gaussian data}}\label{gData}
  \end{center}
\end{table}

\begin{figure}[ht!]
\centering
\includegraphics[width=0.45\textwidth]{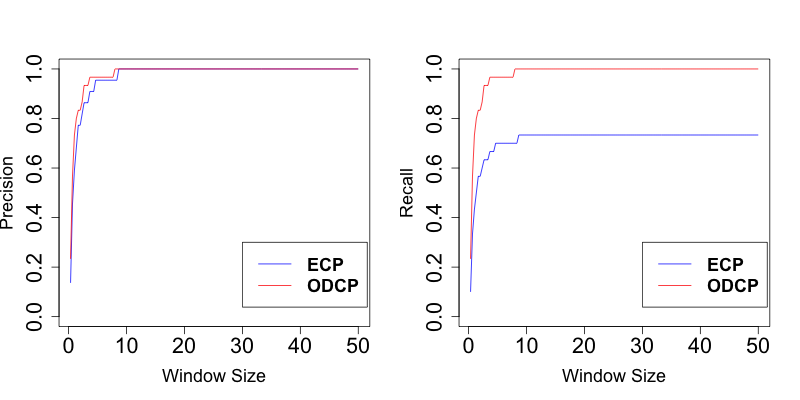}
\caption{\small Effect on performance for variance change on non-compositional data.}
\label{varianceSparsityNonComp}
\end{figure}

\subsubsection{\bf Discussion of Results for Real Data}
In this section, we discuss the performance of ODCP on a Human activity real Dataset $\mathcal{G}_2$ \cite{ermes2008detection}. 
This dataset has 19 activities performed by eight subjects (4 female, 4 male, between the ages 20 and 30) for 5 minutes sampled at 25Hz. Each sample is a 45 dimensional vector and there were totally 7500 samples per subject. The subjects are asked to perform the activities in their own style and were not restricted on how the activities should be performed. For this reason, there are inter-subject variations in the speeds and amplitudes of some activities. The data were recorded from a single subject using five Xsens MTx units attached to the torso, arms and legs.

We generated multivariate data from the underlying dataset, with varying activities and varying number of segments. We consider two experiments using this dataset. In the first experiment ($\mathcal{E}_1$), we randomly chose one of the eight subjects and randomly chose the number of segments, between two and 10 and segment length randomly from 200 to 300 out of 7500 possible samples. Then, for each segment, we chose an activity uniformly at random from among the 19 possible activities. For the second experiment ($\mathcal{E}_2$), we randomly chose an activity from one of the 19 activities. We chose the number of segments and segment length as in $\mathcal{E}_1$. Then, for each segment, we chose a subject performing the chosen activity uniformly at random from among the 8 possible subjects. In both the experimental settings, ODCP detects most of the change points (see Table \ref{hct}).\smallskip

\begin{table}[h]
\begin{center}

\begin{tabular}{|c|c|c|}\hline
Datasets &  \multicolumn{2}{|c|}{ODCP} \\
   \cline{2-3}
   & Prec. & Rec. \\
   \cline{1-3}
$\mathcal{E}_1$ &   0.9 & 0.93 \\
\hline
$\mathcal{E}_2 $ &  0.88 & 1.0 \\ \hline
\end{tabular}
\caption{\small{Performance of ODCP on human activity dataset with window size 4\%}}\label{hct}
\end{center}
\end{table}

\section{Conclusions}\label{conc}
In this paper, we presented an online change point algorithm for multivariate data. Our algorithm ODCP is practical, out-of-the-box and reliable in detecting distributional changes of varied nature (like changes in mean, variance, covariance etc). The experiments on compositional data indicate the flexibility of Dirichlet family in approximating the underlying data distribution.  Furthermore, ODCP is a viable solution for change-point detection in general setting, while being undoubtedly superior in the case of compositional data. The ``nice" properties of the inverse multinomial logit mapping we have considered have seamlessly allowed us to extend the benefits of Dirichlet modeling in identifying change points in compositional data to general scenario. 


%
\bibliographystyle{99}

\end{document}